\documentclass{llncs}

\usepackage{amsmath}
\usepackage{amssymb}
\usepackage{times}
\usepackage{txfonts}
\usepackage{stmaryrd}
\usepackage{url}
\usepackage{amsfonts}
\usepackage{code}
\usepackage{mathrsfs}
\DeclareMathAlphabet{\mathpzc}{OT1}{pzc}{m}{it}
\let\mathpzc\mathscr
\let\mathpzc\mathcal

\def\bondi{{\bf bondi}}
\def\ltrans{[\![}
\def\rtrans{]\!]}

\usepackage{prooftree}
\usepackage{macros}

\renewcommand{\rew}{\redar}

\newcommand{\withrw}[1]{#1}

\title{Expressiveness via Intensionality and Concurrency}
\titlerunning{Expressiveness via Intensionality and Concurrency}
\author{Thomas Given-Wilson
\thanks{This work has been partially supported by the project ANR-12-IS02-001 PACE.}}
\authorrunning{T. Given-Wilson}
\institute{INRIA, Paris, France\\
\email{thomas.given-wilson@inria.fr}}

\begin{document}
\makeatactive

\makeatletter
\def \rightarrowfill{\m@th\mathord{\smash-}\mkern-6mu%
  \cleaders\hbox{$\mkern-2mu\mathord{\smash-}\mkern-2mu$}\hfill
  \mkern-6mu\mathord\to}
\makeatother
\makeatletter
\def \Rightarrowfill{\m@th\mathord{\smash=}\mkern-6mu%
  \cleaders\hbox{$\mkern-2mu\mathord{\smash=}\mkern-2mu$}\hfill
  \mkern-6mu\mathord\Rightarrow}
\makeatother
\def \overstackrel#1#2{\mathrel{\mathop{#1}\limits^{#2}}}

\maketitle  

\vspace{-0.4cm}

\begin{abstract}
Computation can be considered by taking into account two dimensions:
extensional versus intensional,
and sequential versus concurrent.
Traditionally sequential extensional computation can be captured by
the $\lambda$-calculus. However, recent work shows that there are more
expressive intensional calculi such as $SF$-calculus.
Traditionally process calculi capture computation by encoding the 
$\lambda$-calculus, such as in the $\pi$-calculus.
Following this increased expressiveness via intensionality, other
recent work has shown that concurrent pattern calculus is more
expressive than $\pi$-calculus.
This paper formalises the relative expressiveness of all four of these
calculi by placing them on a square whose edges are irreversible
encodings.
This square is representative of a more general result:
that expressiveness increases with both intensionality and concurrency.
\end{abstract} 


\section{Introduction}
\label{sec:intro}

\vspace{-0.1cm}

Computation can be characterised in two dimensions:
{\em extensional} versus {\em intensional};
and {\em sequential} versus {\em concurrent}.
Extensional sequential computation models are those whose {\em functions} cannot distinguish the internal structure of their {\em arguments}, here characterised by the $\lambda$-calculus \cite{Barendregt85}.
However, Jay \& Given-Wilson show that $\lambda$-calculus does not support all sequential computation \cite{JayGW11}.
In particular, there are intensional Turing-computable functions, characterised by {\em pattern-matching}, that can be represented within $SF$-calculus but not within $\lambda$-calculus \cite{JayGW11}.
Of course $\lambda$-calculus can encode Turing computation, but this is a weaker claim.
Ever since Milner et al.~showed that the $\pi$-calculus generalises $\lambda$-calculus \cite{90426,milner.parrow.ea:calculus-mobile},
concurrency theorists expect process calculi to subsume sequential computation as represented by $\lambda$-calculus \cite{90426,milner.parrow.ea:calculus-mobile,citeulike:500640}.
Following from this, here extensional concurrent computation is characterised by process calculi that
do not communicate terms with internal structure, and, at least, support 
$\lambda$-calculus.
Intensional concurrent computation is represented by process calculi whose communication
includes terms with internal structure, and reductions that depend upon the internal structure
of terms.
Here intensional concurrent computation is demonstrated by {\em concurrent pattern calculus} (CPC) that not only generalises intensional pattern-matching from sequential computation to {\em pattern-unification} in a process calculus, but also increases the {\em symmetry} of interaction \cite{GivenWilsonGorlaJay10,givenwilson:hal-00987578}.
 
These four calculi form the corners of a {\em computation square}
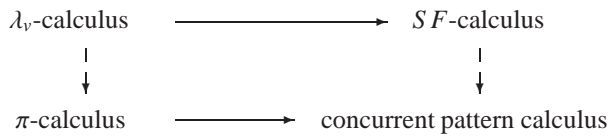
\vspace*{-0.3cm}
\begin{center}
\begin{picture}(250,65)(0,0)
\put(5,43){\mbox{$\lambda_v$-calculus}}
\put(125,43){\mbox{\ \ \ \ \ \ \ \ \ \ \ \ \ $SF$-calculus}}
\put(5,5){\mbox{~$\pi$-calculus}}
\put(123,5){\mbox{concurrent pattern calculus}}
\put(68,45){\vector(1,0){80}}
\put(68,8){\vector(1,0){45}}
\put(34,35){\line(0,-1){5}}
\put(34,26){\vector(0,-1){8}}
\put(183,35){\line(0,-1){5}}
\put(183,26){\vector(0,-1){8}}
\end{picture}
\end{center}
where the left side is merely extensional and the right side also intensional;
the top edge is sequential and the bottom edge concurrent.
All arrows are defined via valid encodings \cite{G:CONCUR08}.
The horizontal (solid) arrows are {\em homomorphisms} in that they also preserve {\em application} or {\em parallel composition}.
The vertical (dashed) arrows are {\em parallel encodings} in that application is mapped to a parallel composition (with some machinery).
Each arrow represents increased expressive power with CPC completing the square.

This paper presents the formalisation of these expressiveness results for the four calculi above.
This involves adapting some popular definitions of encodings \cite{G:IC08,G:DC10,G:CONCUR08}
and then building upon various prior results
\cite{Curry58combinatorylogic,90426,milner.parrow.ea:calculus-mobile,GivenWilsonGorlaJay10,JayGW11,GivenWilsonPHD}.
These can be combined to yield the new expressiveness results here captured by the
computation square.

The organisation of the paper is as follows.
Section~\ref{sec:encoding} reviews prior definitions of encodings and defines the ones used in this paper.
Section~\ref{sec:sec} reviews $\lambda$-calculus and combinatory logic while introducing common definitions.
Section~\ref{sec:sic} summarises intensionality in the sequential setting and formalises the arrow across the top of the square.
Section~\ref{sec:cec} begins concurrency through $\pi$-calculus and its parallel encoding of $\lambda_v$-calculus.
Section~\ref{sec:cic} recalls concurrent pattern calculus and completes the results of the computation square.
Section~\ref{sec:conclusions} draws conclusions, \withrw{considers related work,} and discusses future work.

\section{Encodings}
\label{sec:encoding}

This section recalls valid encodings \cite{G:CONCUR08}
for formally relating process calculi and adapts the definition to define homomorphisms and
parallel encodings.
The validity of valid encodings in developing expressiveness studies emerges from the
various works \cite{G:IC08,G:DC10,G:CONCUR08}, that have also recently inspired similar works
\cite{LPSS10,LVF10,gla12}.
Here the adaptations are precise definitions of homomorphisms that give stronger positive
results (the negative results are not required to be as strong).
Also, parallel encodings are defined to
account for the mixture of sequential and concurrent languages considered.

An {\em encoding} of a language $\Lang_1$ into another language $\Lang_2$ is a pair
$(\encode\cdot,\renpol)$ where $\encode\cdot$ translates every $\Lang_1$-term into
an $\Lang_2$-term and $\renpol$ maps every name (of the source language) into a tuple
of $k$ names (of the target language), for $k > 0$.
The translation $\encode\cdot$ turns every term of the source language into a term of the
target; in doing this, the translation may fix some names to play a precise r\^ole 
or may translate a single name into a tuple of names. This can be obtained
by exploiting $\renpol$.

Now consider only encodings that satisfy the following properties.
Let a {\em $k$-ary context} $\context C {\_\,_1; \ldots; \_\,_k}$ be a term with $k$
holes $\{\_\,_1;\ldots;\_\,_k\}$ that appear exactly once each.
Moreover, denote with $\redar_i$ and $\Redar_i$ 
the relations $\redar$ (reduction relation) and
$\Redar$ (the reflexive transitive closure of $\redar$) in language $\Lang_i$;
denote with $\redar^\omega_i$ an infinite sequence of reductions in $\Lang_i$.
Moreover,
let $\equiv_i$ denote the structural equivalence relation for a language $\Lang_i$, and
$\bisim_i$ denote the reference behavioural equivalence for language $\Lang_i$.
For simplicity the notation $T \redar_i \equiv_i T'$ denotes that there exists $T''$
such that $T\redar_i T''$ and $T''\equiv_i T'$, and may also be used with $\Redar_i$
or $\bisim_i$. 
Also, let $P \suc_i$ mean that there exists $P'$ such that $P \Redar_i P'$ and $P' \equiv_i P''\bnf \ok$,
for some $P''$ where $\ok$ is a specific process to indicate success.
Finally, to simplify reading, let $S$ range
over terms of the source language (viz., $\Lang_1$) and $T$ range
over terms of the target language (viz., $\Lang_2$).

\begin{definition}[Valid Encoding (from \cite{G:CONCUR08})]
\label{def:ve}
An encoding $(\encode\cdot,\renpol)$ of $\Lang_1$ into $\Lang_2$
is {\em valid} if it satisfies the following five properties:
\begin{enumerate}
\item {\em Compositionality:} for every $k$-ary operator $\op$ of $\Lang_1$
and for every subset of names $N$,
there exists a $k$-ary context $\CopN C \op N {\_\,_1; \ldots; \_\,_k}$ of $\Lang_2$
such that, for all $S_1,\ldots,S_k$ with ${\sf fn}(S_1,\ldots,S_k) = N$, it holds
that $\encode{\op(S_1,\ldots,S_k)} = \CopN C \op N {\encode{S_1};\ldots;\encode{S_k}}$.

\item {\em Name invariance:}
for every $S$ and name substitution $\sigma$, it holds that
$$
\encode{\sigma S}\ \left\{ 
\begin{array}{ll}
\ =\ \sigma'\encode S& \mbox{ if $\sigma $ is injective}\\
\ \bisim_2\ \sigma'\encode S  & \mbox{ otherwise}
\end{array}
\right.
$$
where $\sigma'$ is such that 
$\renpol(\sigma(a)) = \sigma'(\renpol(a))$
for every name $a$.

\item {\em Operational correspondence:}
\begin{itemize}
\item for all $S \Redar_1 S'$, it holds that $\encode S \Redar_2 \bisim_2 \encode {S'}$;
\item for all $\encode S \Redar_2 T$, there exists $S'$ such that $S \Redar_1\!\! S'$ 
and $T \Redar_2 \bisim_2\!\! \encode {S'}$.
\end{itemize}

\item {\em Divergence reflection:}
for every $S$ such that 
$\encode S \redar\!\!_2^\omega$, it holds that 
\linebreak $S$ \mbox{$\redar\!\!_1^\omega$}.

\item {\em Success sensitiveness:}
for every $S$, it holds that $S \suc_1$ if and only if $\encode S \suc_2$.
\end{enumerate}
\end{definition}

Observe that the definition of valid encoding is very general and, with the exception
of success sensitiveness, can apply to sequential languages such as $\lambda$-calculus
as well as process calculi.
(On the understanding that a name substitution for sequential
calculi is a mapping from names/variables to names/variables {\em not} terms.)
However, the relations presented in this work bring together a variety of
prior results and account for them in a stronger and more uniform manner.
To this end, the following definitions support the results.
The first two define homomorphism in the sequential and concurrent settings.

\begin{definition}[Homomorphism (Sequential)]
\label{def:homo-seq}
A {\em (sequential) homomorphism} is
a translation $\encode\cdot$ from one language to another that satisfies:
compositionality, name invariance, operational correspondence, and divergence reflection;
and that preserves application, i.e.~where $\encode {S_1~S_2} = \encode{S_1}~\encode{S_2}$.
\end{definition}

\begin{definition}[Homomorphism (Concurrent)]
\label{def:homo-con}
A {\em (concurrent) homomorphism} is a valid encoding whose translation preserves parallel composition,
i.e.~$\encode {P_1\bnf P_2}=\encode {P_1}\bnf \encode {P_2}$.
\end{definition}

The next is for encoding sequential languages into concurrent languages and
exploits
that $\encode\cdot _c$ indicates
an encoding from source terms to target terms that is parametrised by a name $c$.

\begin{definition}[Parallel Encoding]
\label{def:pe}
An encoding $(\encode\cdot_c,\renpol)$ of $\Lang_1$ into $\Lang_2$
is a {\em parallel encoding} if it satisfies the first four properties of a valid encoding
(compositionality, name invariance, operational correspondence, 
and divergence reflection)
and the following additional property.
\begin{enumerate}
\setcounter{enumi}{4}
\item {\em Parallelisation:} 
      The translation of the application $MN$ is of the form
      $\xtrans {MN} _c \define \res {n_1} \res {n_2} (
        {\mathcal A}(c,n1,n2)\bnf \xtrans M _{n1} \bnf \xtrans N _{n2})$
      where ${\mathcal A}$ is a process parametrised by $c$ and $n1$ and $n2$.
\end{enumerate}
\end{definition}

Parallelisation is a restriction on the more general compositionality criteria.
Here this ensures that in addition to compositionality, the translation must allow
for independent reduction of the components of an application.
As the shift from sequential to concurrent computation can exploit this to support parallel
reductions, the definition of parallel encoding 
encourages more flexibility in reduction since components can be reduced independently.

The removal of the success sensitiveness property is for simplicity when using prior
results.
It is not difficult to include success sensitiveness, this involves adding the success
primitive to the sequential languages and defining $S\suc$, e.g.~$S\suc$ means that
$S\redar^*\ok$.
Additionally, this requires adding a test process $Q_c$ to the definition of parallel encoding
with success sensitiveness defined by:
``for every $S$, it holds that $S \suc_1$ if and only if $\encode S _c \bnf Q_c\suc_2$.
However, since adding the success state $\ok$ to $\l$-calculus and combinatory
logics\footnote{The results for intensional combinatory logics require that success
behaves as a {\em constructor} as discussed for various combinatory logics in \cite{JayGW11}.}
would require redoing many existing results, it is easier to avoid the added
complexity since no clarity or gain in significance is made by adding it.

Encodings from concurrent languages into sequential ones have not been defined specifically
here since they prove impossible. The proof of these results relies merely on the
requirement of 
operational correspondence, and so shall be
done on a case-by-case basis.

\section{Sequential Extensional Computation}
\label{sec:sec}

Both $\lambda$-calculus and traditional combinatory logic base reduction rules upon the application of a function to one or more arguments.
Functions in both models are extensional in nature, that is a function does not have direct access to the internal structure of its arguments.
Thus, functions that are extensionally equal are indistinguishable within either model even though they may have different normal forms.

The relationship between the $\lambda$-calculus and traditional combinatory logic is closer than sharing application-based reduction and extensionality.
There is a homomorphism from call-by-value $\lambda_v$-calculus into any combinatory logic that supports the combinators $S$ and $K$ \cite{Curry58combinatorylogic,Barendregt85}.
There is also a homomorphism from traditional combinatory logic to a $\lambda$-calculus with more generous operational semantics \cite{Curry58combinatorylogic,Barendregt85}.

\subsection{$\lambda$-calculus}
\label{ssec:lambda}

The {\em term} syntax of the $\lambda$-calculus is given by
\[
t ::= x \bnf t~t \bnf \lambda x.t \; .
\]
The {\em free variables} of a term are defined in the usual manner.
A {\em substitution} $\sigma$ is defined as a partial function from variables to terms.
The {\em domain} of $\sigma$ is denoted $\textsf{dom}(\sigma)$;
the free variables of $\sigma$, written $\textsf{fv}(\sigma)$, is given by the union of the sets $\textsf{fv}(\sigma x)$ where $x \in \textsf{dom}(\sigma)$.
The {\em variables} of $\sigma$, written $\textsf{vars}(\sigma)$, are $\textsf{dom}(\sigma)\cup\textsf{fv}(\sigma)$.
A substitution $\sigma$ {\em avoids} a variable $x$ (or collection of variables $\mu$) if $x\notin\textsf{vars}(\sigma)$ (respectively $\mu\cap\textsf{vars}(\sigma)=\{\}$).
Note that all substitutions considered in this paper have finite domain.
The application of a substitution $\sigma$ to a term $t$ is defined as usual, as is $\alpha$-conversion $=_\alpha$.

There are several variations of the $\lambda$-calculus with different operational semantics.
For construction of the computation square by exploiting the results of Milner et al.~\cite{90426}, it is necessary to choose an operation semantics, such as {\em call-by-value} $\lambda_v$-calculus or {\em lazy} $\lambda_l$-calculus.
The choice here is to use call-by-value $\lambda_v$-calculus, although the results can be reproduced for lazy $\lambda_l$-calculus as well.
In addition a more generous operation semantics for $\lambda$-calculus will be presented for later discussion and relations.

To formalise the reduction of call-by-value $\lambda_v$-calculus requires a notion of {\em value} $v$. These are defined in the usual way, by
\begin{eqnarray*}
v &::=& x \bnf \lambda x.t
\end{eqnarray*}
consisting of variables and $\lambda$-abstractions.

Computation in the $\lambda_v$-calculus is through the $\beta_v$-reduction rule
\begin{eqnarray*}
(\lambda x.t)v &\quad\rew_v\quad& \{v/x\}t \; .
\end{eqnarray*}
When an abstraction $\lambda x.t$ is applied to a value $v$ then substitute $v$ for $x$ in the body $t$.
The {\em reduction relation} (also denoted $\rew_v$) is the smallest that satisfies the following rules
\begin{equation*}
\Rule{}{}{(\lambda x.t)v \ \rew_v\  \{v/x\}t}{}
\qquad
\Rule{}{s\rew_v s'}{s~t\rew_v s'~t}{}
\qquad
\Rule{}{t\rew_v t'}{s~t\rew_v s~t'}{} \; .
\end{equation*}
The transitive closure of the reduction relation is denoted $\rew_v^*$ though the star may be elided if it is obvious from the context.

The more generous operational semantics for the $\lambda$-calculus allows any term to be the argument when defining $\beta$-reduction. Thus the more generous $\beta$-reduction rule is
\begin{eqnarray*}
(\lambda x.s)t \quad\rew\quad \{t/x\}s 
\end{eqnarray*}
where $t$ is any term of the $\lambda$-calculus.
The reduction relation $\rew$ and the transitive closure thereof $\rew^*$ are obvious adaptations from those for the $\lambda_v$-calculus.
Observe that any reduction $\rew_v$ of $\lambda_v$-calculus is also a reduction $\rew$ of $\lambda$-calculus.

\subsection{Traditional Combinatory Logic}
\label{ssec:sk}

A {\em combinatory calculus} is given by a finite collection ${\mathcal
  O}$ of {\em operators} (meta-variable $O$) that are
used to define the ${\mathcal O}$-{\em combinators} (meta-variables
$M,N,X,Y,Z$) built from these by application
\[
M,N ::= O \bnf MN \; .
\]
The {\em ${\mathcal O}$-combinatory calculus} or {\em ${\mathcal O}$-calculus}
is given by the combinators plus their reduction rules.

Traditional combinatory logic can be represented by two combinators $S$ and $K$ \cite{Curry58combinatorylogic} so the $SK$-calculus has {\em reduction rules}
\[
\begin{array}{rcll}
SMNX &\rew& MX(NX)\\
KXY &\rew& X \; .
\end{array}
\]
The combinator $SMNX$ duplicates $X$ as the argument to both $M$ and
$N$.  The combinator $KXY$ eliminates $Y$ and returns $X$.
The {\em reduction relation} $\rew$ is 
as for $\lambda$-calculus.

Although this is sufficient to provide a direct account of functions in the style
of $\lambda$-calculus, an alternative is to consider the representation
of arbitrary computable functions that act upon combinators.

A {\em symbolic function} is defined to be an $n$-ary partial
function ${\mathcal G}$ of some combinatory logic, i.e.~a function of
the combinators that preserves their equality, as determined by the
reduction rules.  That is, if $X_i = Y_i$ for $1\le i\le n$ then
${\mathcal G}(X_1,X_2,\ldots,X_n) = {\mathcal G}(Y_1,Y_2,\ldots, Y_n)$
if both sides are defined.  A symbolic function is {\em restricted} to
a set of combinators, e.g.\ the normal forms, if its domain is within
the given set.

A combinator $G$ in a calculus {\em represents} ${\mathcal G}$ if
\[
GX_1\ldots X_n = {\mathcal G}(X_1,\ldots,X_n)
\]
whenever the right-hand side is defined.  For example, the symbolic
functions
${\mathcal S}(X_1,X_2,X_3) = X_1X_3(X_2X_3)$ and
${\mathcal K}(X_1,X_2) = X_1$
are represented by $S$ and $K$, respectively, 
in $SK$-calculus.
  Consider the
symbolic function
${\mathcal I}(X) = X$. 
In $SKI$-calculus where $I$ has the rule
$IY \rew Y$
then ${\mathcal I}$ is represented by $I$. In both $SKI$-calculus and
$SK$-calculus, ${\mathcal I}$ is represented by any combinator of the
form $SKX$ since
\[
SKXY \redar KY(XY) \redar Y\;.
\]
For convenience define the {\em identity combinator} $I$ in $SK$-calculus to be $SKK$.

\subsection{Relations}
\label{sec:lambda=SK}

One of the goals of combinatory logic is to give an equational account
of variable binding and substitution, particularly as it appears in
$\lambda$-calculus.
In order to represent $\lambda$-abstraction, it is necessary to have some
variables to work with. Given ${\mathcal O}$ as before, define the
{\em ${\mathcal O}$-terms} by 
\[
M,N ::= x \bnf O \bnf MN 
\]
where $x$ is as in $\lambda$-calculus.
Free variables, substitutions, and symbolic computations are defined
just as for ${\mathcal O}$-calculus.

Given a variable $x$ and term $M$ define a symbolic function ${\mathcal G}$ on terms by 
\[
{\mathcal G}(X) = \{X/x\}M\; .
\]
Note that if $M$ has no free variables other than $x$ then ${\mathcal G}$
is also a symbolic computation of the combinatory logic.  If every
such function ${\mathcal G}$ on ${\mathcal O}$-combinators is representable then the
${\mathcal O}$-combinatory logic is {\em combinatorially complete} in the
sense of Curry \cite[p.~5]{Curry58combinatorylogic}.

Given $S$ and $K$ then ${\mathcal G}$ above can be
represented by a term $\l^* x.M$ given by
\begin{equation*}
\begin{array}{rcl}
\lambda^{*}x.x &=& I\\
\lambda^{*}x.y &=& K y \quad\mbox{if $y\neq x$}
\end{array}\qquad\qquad
\begin{array}{rcl}
\lambda^{*}x.O &=& KO  \\
\lambda^{*}x.M N &=& S (\lambda^{*}x.M) (\lambda^{*}x.N) \; .
\end{array}
\end{equation*}

The following lemmas are central results of combinatory logic
\cite{Curry58combinatorylogic} and Theorem~2.3 of \cite{JayGW11}. This is sufficient
to show there is a homomorphism from $\lambda_v$-calculus to any combinatory calculus 
that represents $S$ and $K$.

\begin{lemma}
\label{lem:beta*}
For all terms $M$ and $N$ and variables $x$ there is a reduction
$(\lambda ^*x.M)~N \rew^* \{N/x\}M$.
\end{lemma}

\begin{lemma}
\label{lem:lambda2SK}
  Any combinatory calculus that is able to represent $S$ and $K$
  is combinatorially complete.
\end{lemma}

\begin{theorem}
\label{thm:lambda2SK}
There is a homomorphism (Definition~\ref{def:homo-seq}) from $\lambda$-calculus into $SK$-calculus.
\end{theorem}
\begin{proof}
Compositionality, name invariance, and preservation of application hold by construction.
Operational correspondence and divergence reflection can by proved via Lemma~\ref{lem:lambda2SK}.
\end{proof}

Below is a standard translation from $SK$-calculus into $\lambda$-calculus that
preserves reduction and supports the following lemma \cite{Curry58combinatorylogic,Barendregt85}.
\begin{eqnarray*}
\ltrans S\rtrans = \lambda g.\lambda f.\lambda x.g~x~(f~x)&\qquad\quad
\ltrans K\rtrans = \lambda x.\lambda y.x\qquad\quad&
\ltrans MN \rtrans = \ltrans M\rtrans~\ltrans N\rtrans
\end{eqnarray*}

\begin{lemma}[Theorem~2.3.3 of \cite{GivenWilsonPHD}]
\label{lem:SKinLambda}
Translation from $SK$-calculus to $\lambda$-calculus preserves the reduction relation.
\end{lemma}

\begin{theorem}
\label{thm:SK2lambda}
There is a homomorphism (Definition~\ref{def:homo-seq}) from $SK$-calculus into $\lambda$-calculus.
\end{theorem}
\begin{proof}
Compositionality and preservation of application hold by construction. Name invariance is trivial.
Operational correspondence and divergence reflection are proved via Lemma~\ref{lem:SKinLambda}.
\end{proof}

Although the top left corner of the computation square is populated by $\lambda_v$-calculus, the arrows out allow for either $\lambda_v$-calculus or $SK$-calculus to be used.
Indeed, the homomorphisms in both directions between $\lambda$-calculus and $SK$-calculus allow these two calculi to be considered equivalent.

\section{Sequential Intensional Computation}
\label{sec:sic}

Intuitively intensional functions are more expressive than merely extensional functions,
however populating the top right corner of the computation square requires more formality than intuition.
The cleanest account of this is by considering combinatory logic.

Even in $SK$-calculus there are Turing-computable functions defined upon the combinators that cannot be represented within $SK$-calculus.
For example, consider the function that reduces any combinator of the form $SKX$ to $X$.
Such a function cannot be represented in $SK$-calculus, or $\lambda$-calculus, as all combinators of the form $SKX$ represent the identity function. However, such a function is Turing-computable and definable upon the combinators.
This is an example of a more general problem of {\em factorising} combinators that are both applications and stable under reduction.

Exploiting this factorisation is $SF$-calculus \cite{JayGW11} that is able to support intensional functions on combinators including a structural equality of normal forms.
Thus $SF$-calculus sits at the top right hand corner of the computation square.
The arrow across the top of the square is formalised by showing a homomorphism from $SK$-calculus into $SF$-calculus. The lack of a converse has been proven by showing that the intensionality of $SF$-calculus cannot be represented within $SK$-calculus, or $\lambda$-calculus \cite{JayGW11}.

\subsection{Symbolic Functions}
\label{sec:SF-sym}

Symbolic functions need not be merely extensional, indeed it is possible to define symbolic functions that consider the structure of their arguments.
Observe that each operator $O$ has an {\em arity} given by the minimum number of arguments it requires to
instantiate a rule. Thus, $K$ has arity $2$ while $S$ has arity $3$. A
{\em partially applied operator} is a combinator of the form
$OX_1\ldots X_k$ where $k$ is less than the arity of $O$.  An operator
with a positive arity is an {\em atom} (meta-variable $A$).  A
partially applied operator that is an application is a {\em compound}.
Hence, the partially applied operators of $SK$-calculus are the atoms
$S$ and $K$, and the compounds $SM$, $SMN$ and $KM$ for any $M$ and
$N$.

Now define a {\em factorisation function} ${\mathcal F}$ on
combinators by
\[
\begin{array}{rcll}
{\mathcal F}(A,M,N) &\rew& M &\mbox{if $A$ is an atom}\\
{\mathcal F}(XY,M,N) &\rew& NXY \quad&\mbox{if $XY$ is a compound.} 
\end{array}
\]

\begin{lemma}[Theorem~3.2 of \cite{JayGW11}]
\label{lem:noFinSK:1}
Factorisation of $SK$-combinators is a symbolic computation that is
not representable within $SK$-calculus.
\end{lemma}
\begin{proof}
  Suppose that there is an $SK$-combinator $F$ that represents
  ${\mathcal F}$.  Then, for any combinator $X$ it follows that
$F(SKX)S(KI) \rew  KI(SK)X  \rew X$. 
Translating this to $\l$-calculus as in Lemma~\ref{lem:beta*} yields
$\ltrans F(SKX)S(KI) \rtrans \rew \ltrans X\rtrans$ and also
$\ltrans F(SKX)S(KI) \rtrans 
=  \ltrans F \rtrans ~\ltrans (SKX)\rtrans~\ltrans S\rtrans~\ltrans KI\rtrans 
\rew \ltrans F \rtrans~ (\l x.x)~\ltrans S\rtrans~\ltrans KI\rtrans$. 
Hence, by confluence of reduction in $\l$-calculus, all $\ltrans
X\rtrans$ share a reduct with $\ltrans F \rtrans~ (\l x.x)~\ltrans
S\rtrans~\ltrans KI\rtrans$ but this is  impossible since
$\ltrans S\rtrans$ and $\ltrans K\rtrans$ are distinct normal forms.
Hence ${\mathcal F}$ cannot be represented by an  $SK$-combinator. 
\end{proof}

\subsection{$SF$-calculus}
\label{ssec:computation:sec:factor}

When considering intensionality in a combinatory logic it is tempting to specify
a factorisation combinator $F$ as a representative for ${\mathcal F}$.
However, ${\mathcal F}$ is defined using partially applied operators,
which cannot be known until all reduction rules are given,
including those for $F$.
This circularity of definition is broken by beginning with a syntactic
characterisation of the combinators that are to be factorable.
 
The $SF$-calculus \cite{JayGW11} has {\em factorable forms} given by
$S \bnf  SM \bnf SMN \bnf F \bnf FM \bnf FMN$
and {\em reduction rules} 
\[
\begin{array}{rcll}
  SMNX &\rew& MX(NX)\\
  FOMN &\rew& M &\mbox{if $O$ is $S$ or $F$}\\
  F(XY)MN &\rew& NXY &\mbox{if $XY$ is a factorable form.}
\end{array}
\]

The expressive power of $SF$-calculus subsumes that of $SK$-calculus
since $K$ is here defined to be $FF$ and $I$ is defined to be $SKK$
as before.

\begin{lemma}
\label{lem:SK2SF}
There is a homomorphism (Definition~\ref{def:homo-seq}) from $SK$-calculus into $SF$-calculus.
\end{lemma}

\begin{theorem}
\label{thm:lambdav2SF}
There is a homomorphism (Definition~\ref{def:homo-seq}) from $\lambda_v$-calculus to $SF$-calculus.
\end{theorem}
\begin{proof}
By Theorem~\ref{thm:lambda2SK} and Lemma~\ref{lem:SK2SF}.
\end{proof}

\begin{lemma}
\label{lem:noSF2lambda}
There is no reduction preserving translation $\encode\cdot$ from $SF$-calculus to $\lambda_v$-calculus.
\end{lemma}
\begin{proof}
By Lemma~\ref{lem:noFinSK:1}.
\end{proof}

\begin{theorem}
\label{thm:noSF2lambda}
There is no homomorphism (Definition~\ref{def:homo-seq}) from $SF$-calculus to $\lambda_v$-calculus.
\end{theorem}
\begin{proof}
Lemma~\ref{lem:noSF2lambda} shows that operational correspondence is impossible.
\end{proof}

This completes the top edge of the computation square by showing that
$SF$-calculus subsumes $\lambda_v$-calculus and that the subsumption is
irreversible.
Indeed, these results hold for $\lambda$-calculus \cite[Theorem~5.2.6]{GivenWilsonPHD} and
$SK$-calculus (by Lemma~\ref{lem:noFinSK:1}) as well.

\section{Concurrent Extensional Computation}
\label{sec:cec}

The bottom left corner of the computation square considers extensional concurrent computation,
here defined to be extensional process calculi that subsume $\lambda$-calculus.
The $\pi$-calculus \cite{milner.parrow.ea:calculus-mobile} holds a pivotal r\^ole amongst process
calculi due to popularity, being the first to represent topological changes, and subsuming
$\lambda_v$-calculus \cite{90426}.
Note that although there are many $\pi$-calculi, the one here is that used by Milner so as to 
more easily exploit previous results \cite{90426} (and here augmented with a success process $\ok$).

The processes for the $\pi$-calculus are given as follows and exploit a class of names (denoted $m,n,x,y,z,\ldots$ similar to variables in the $\lambda$-calculus):
\begin{equation*}
P \quad ::=\quad  \zero \ \bnf\  P\!\bnf\!P \ \bnf\  !P \ \bnf\  \res a P \ \bnf\  \iap a b.P \ \bnf\  \oap a b.P\bnf \ok\; .
\end{equation*}
The names of the $\pi$-calculus are used for channels of communication and for information being communicated. 
The {\em free names} of a process ${\sf fn}(P)$ are as usual. 
{\em Substitutions} in the $\pi$-calculus are partial functions that map names to names,
with domain, range, free names, names, and avoidance, all straightforward adaptations from
substitutions of the $\lambda$-calculus.
The application of a substitution to a process is defined in the usual manner.
Issues where substitutions must avoid restricted or input names are handled by $\alpha$-conversion $=_\alpha$ that is the congruence relation defined in the usual manner.
The general {\em structural equivalence relation} $\equiv$ is defined by:
\begin{eqnarray*}
&P\bnf \stoppr \equiv P \qquad
P \bnf Q \equiv Q \bnf P\qquad
P \bnf (Q \bnf R) \equiv (P \bnf Q) \bnf R&\\&
!P \equiv P\bnf !P\qquad
\res n \stoppr \equiv \stoppr\qquad
\res n \res m P \equiv \res m \res n P&\\&
P\, | \, \res n Q \equiv \res n (P\, |\, Q)
\ \ \mbox{if } n \!\not\in\! {\sf fn}(P)&
\end{eqnarray*}

The $\pi$-calculus has one {\em reduction rule} given by
\begin{eqnarray*}
\iap a b.P \bnf \oap a c.Q &\quad\redar\quad& \{c/b\}P\bnf Q \; .
\end{eqnarray*}
The reduction rule is then closed under parallel composition, restriction and structural equivalence to yield the reduction relation $\redar$ as follows:
\begin{eqnarray*}
\begin{array}{c}
\Rule{}
{P \redar P'}{P\bnf Q \redar P'\bnf Q}{}\qquad
\Rule{}
{P \redar P'}{\res n P \redar \res n P'}{}\qquad
\Rule{}
{P \equiv Q \quad Q \redar Q'\quad Q' \equiv P'}{P \redar P'}{} \; .
\end{array}
\end{eqnarray*}

Now that the $\pi$-calculus and process calculus concepts are recalled, it remains to
demonstrate that Milner's encoding \cite{90426} can meet the criteria for a parallel encoding.
As the $\beta_v$-reduction rule depends upon the argument being a value the translation into $\pi$-calculus must be able to recognise values. Thus, Milner defines the following
\begin{eqnarray*}
\xtrans {y:=\lambda x.t} \ \define\  !y(w).w(x).w(c).\xtrans t _c&\qquad&
\xtrans {y:=x} \ \define\ !y(w).\oap x w \; .
\end{eqnarray*}
Also the following translation of $\lambda_v$-terms
\begin{equation*}
\begin{array}{rcll}
\xtrans {v} _c &\define& \res y \oap c y .\xtrans {y:=v} & \mbox{$y$ not free in $v$}\\
\xtrans {s~t} _c &\define& \res q \res r ({\sf ap}(c,q,r) \bnf \xtrans s _q \bnf \xtrans t _r)\\
{\sf ap}(p,q,r) &\define& \iap q y .\res v \oap y v .\iap r z .\oap v z .\oap v p \; .
\end{array}
\end{equation*}

\begin{lemma}
\label{lem:l2pi-red}
The translation $\encode \cdot _c$ preserves and reflects reduction.
That is:
\begin{enumerate}
\item If $s\redar_v t$ then $\encode s _c \Redar\bisim \encode t _c$;
\item if $\encode s _c \redar Q$ then there exists $Q'$ and $s'$ such that
      $Q\Redar Q'$ and $Q'\bisim \encode {s'} _c$ and either $s\redar_v s'$ or $s=s'$.
\end{enumerate}
\end{lemma}
\begin{proof}
The first part can be proved by exploiting Milner's Theorem~7.7 \cite{90426}.
The second is by considering the reduction $\encode s _c\redar Q$ which must arise
from the encoding of an application. It is then straightforward to show that either:
the reductions $Q\Redar Q'$ correspond only to translated applications and thus
$Q'\bisim \encode s _c$; or
the reductions are due to a $\lambda_v$-abstraction and thus
$Q'\bisim\encode {s'}_c$ and $s\redar_v s'$.
\end{proof}

\begin{theorem}
\label{thm:lambda2pi}
The translation $\xtrans\cdot _c$ is a parallel encoding (Definition~\ref{def:pe}) from $\lambda_v$-calculus to $\pi$-calculus.
\end{theorem}
\begin{proof}
Compositionality, parallelisation, and name invariance hold by construction.
Operational correspondence follows from Lemma~\ref{lem:l2pi-red}.
Divergence reflection can be proved by observing that the only reductions introduced
in the translation that do not correspond to reductions in the source language are
from translated applications, and these are bounded by the size of the source term.
\end{proof}

There is some difficulty in attempting to define the analogue of a parallel encoding or
homomorphism from a language with a parallel composition operator into a language without.
However, this difficulty can be avoided by observing that any valid encoding, parallel
encoding, or homomorphism must preserve reduction.
Reduction preservation can then be exploited to show when an encoding is impossible.
Here this is by exploiting Theorem~14.4.12 of Barendregt \cite{Barendregt85}, showing that
$\lambda$-calculus is unable to render concurrency or support concurrent computations.

\begin{theorem}
\label{thm:no-pi2lambda}
There is no reduction preserving encoding of $\pi$-calculus into $\lambda$-calculus.
\end{theorem}
\begin{proof}
Define the parallel-or function and show that it can be represented in $\pi$-calculus
but not $\lambda$-calculus.
The parallel-or function is a function $g(x,y)$ that satisfies the following three rules
$g(\bot,\bot) \ \rew^*\ \bot$ and 
$g({\texttt T},\bot) \ \rew^*\ {\texttt T}$ and
$g(\bot,{\texttt T}) \ \rew^*\ {\texttt T}$
where $\bot$ represents non-termination and {\texttt T} represents true.
Such a function is trivial to encode in $\pi$-calculus by
$g(n_1,n_2)= G = \iap {n_1} x.\oap m x.\zero\bnf \iap {n_2} x.\oap m x.\zero$.
Consider $G$ in parallel with two processes $P_1$ and $P_2$ that output their result on $n_1$
and $n_2$, respectively.
If either $P_1$ or $P_2$ outputs {\texttt T} then $G$ will also output {\texttt T} along $m$.
Clearly $\pi$-calculus can represent the parallel-or function, and since Barendregt's
Theorem~14.4.12 shows that $\lambda$-calculus cannot, there cannot be any reduction preserving
encoding of $\pi$-calculus into $\lambda$-calculus.
\end{proof}

\section{Concurrent Intensional Computation}
\label{sec:cic}

Intensionality in sequential computation yields greater expressive power so it is natural to consider intensional concurrent computation.
Intensionality in CPC is supported by a generalisation of pattern-matching to symmetric {\em pattern-unification} that provides the basis for defining interaction.

\subsection{Concurrent Pattern Calculus}
\label{ssec:cpc}

The {\em patterns} (meta-variables $p,p',p_1,q,q',q_1,\ldots$)
are built using a class of {\em names} familiar from $\pi$-calculus and
have the following forms
\begin{eqnarray*}
p &::=& \lambda x \bnf x \bnf \pro x \bnf p\bullet p
\end{eqnarray*}
Binding names $\lambda x$ denote an input sought by the pattern.
Variable names $x$  may be output or tested for equality.
Protected names $\pro x$ can only be tested for equality.
A compound combines two patterns $p$ and $q$, its {\em components}, into a pattern $p\bullet q$ and is left associative.
The {\em atoms} are patterns that are not compounds and the atoms $x$ and $\pro x$ are defined to {\em know} $x$.
The binding names of a pattern must be pairwise distinct.

A {\em communicable} pattern contains no binding or protected names.
Given a pattern $p$, the
binding names ${\sf bn}(p)$,
variable names ${\sf vn}(p)$,
and protected names ${\sf pn}(p)$,
are as expected, with the free names ${\sf fn}(p)$ being the union of variable and protected names.

A {\em substitution} $\sigma$ (also denoted
$\sigma_1,\rho,\rho_1,\theta,\theta_1,\ldots$) is a partial function
from names to communicable patterns.
Otherwise substitutions and their
properties are familiar from earlier sections and are applied to patterns
in the obvious manner.
(Observe that protection can be extended to a communicable pattern by
$\pro {p\bullet q}=\pro p\bullet \pro q$ in the application of a
substitution to a protected name.)

The {\em symmetric matching} or {\em unification} $\{p \pmatch q\}$ of
two patterns $p$ and $q$ attempts to unify $p$ and $q$ by generating
substitutions upon their binding names. When defined, the result is
some pair of substitutions whose domains are the binding names of $p$
and of $q$, respectively.
The rules to generate the substitutions are:
\begin{eqnarray*}
\begin{array}{rcll}
\{x\pmatch x\} =
\{x\pmatch \pro{x}\} &=&
\{\pro{x}\pmatch x\} =
\{\pro{x}\pmatch \pro{x}\} \define & (\{\}, \{\})\\
\{\lambda x\pmatch q\} &\define&  (\{q/x\}, \{\}) & \mbox{if $q$ is communicable}\\
\{p\pmatch \lambda x\} &\define&  (\{\}, \{p/x\}) & \mbox{if $p$ is communicable}\\
\{p_1\bullet p_2\pmatch q_1\bullet q_2\} &\define&  (\sigma_1\cup\sigma_2\,,\, \rho_1\cup\rho_2)
			& \mbox{if } \{p_i\pmatch q_i\}= (\sigma_i, \rho_i) \mbox{ for } i \in \{1,2\}\\
\end{array}
\end{eqnarray*}
Two atoms unify if they know the same name. A binding name unifies with any communicable pattern to
produce a binding for its underlying name.  Two compounds unify if
their corresponding components do; the resulting substitutions are
given by taking unions of those produced by unifying the components.
Otherwise the patterns cannot be unified and the unification is undefined.

The processes of CPC are the same as $\pi$-calculus except the
input and output are replaced by the {\em case} $p\pre P$ with pattern $p$ and body $P$.
A case with the null process as the body $p\to \zero$ may also be written $p$ when no
ambiguity may occur.

The free names of processes, denoted ${\sf fn}(P)$, are defined as
usual for all the traditional primitives and
${\sf fn}(p\pre P) \quad =\quad {\sf fn}(p)\cup({\sf fn}(P)\backslash{\sf bn}(p))$
for the case. As expected the binding names of the pattern bind their free occurrences in the body.
The application $\sigma P$ of a substitution $\sigma$ to a process $P$ is defined in the usual
manner to avoid name capture. For cases this ensures that substitution avoids the binding names in the pattern:
$\sigma(p\to P) = (\sigma p)\to (\sigma P)$ if $\sigma$ avoids ${\sf bn}(p)$.
Renaming via $\alpha$-conversion is defined in the usual manner
\cite{GivenWilsonGorlaJay10,GivenWilsonPHD,givenwilson:hal-00987578}.
The general {\em structural equivalence relation} $\equiv$ is defined
just as in $\pi$-calculus.

\label{sec:reduction-new}
CPC has one {\em interaction axiom} given by
\begin{eqnarray*}
\begin{array}{rcll}
(p\pre P)\bnf(q\pre Q) &\quad\redar\quad& (\sigma P)\bnf(\rho Q) & \qquad\mbox{if $\{p\pmatch q\}=(\sigma,\rho)$} \; .
\end{array}
\end{eqnarray*}
It states that if the unification of two patterns $p$ and $q$ is
defined and generates $(\sigma, \rho)$, then apply the
substitutions $\sigma$ and $\rho$ to the bodies $P$ and $Q$,
respectively.  If the matching of $p$ and $q$ is undefined then no
interaction occurs.
The interaction rule is then closed under
parallel composition, restriction and structural equivalence in the usual manner.
The reflexive and transitive closure of $\redar$ is denoted $\Redar$.
Finally, the reference behavioural equivalence relation $\bisim$ for CPC is already
well detailed \cite{GivenWilsonPHD,GivenWilsonGorla13,givenwilson:hal-00987578}.

\subsection{Completing the Square}
\label{ssec:completing}

Support for both intensionality and concurrency places CPC at the bottom right corner of the computation square.
This section shows how $SF$-calculus and $\pi$-calculus can both be subsumed by CPC, and thus completes the computation square.

Down the right side of the square there is a parallel encoding from $SF$-calculus into CPC that also maps the combinators $S$ and $F$ to reserved names $S$ and $F$, respectively.
The impossibility of finding a parallel encoding of CPC into $SF$-calculus is proved in the same manner as the relation between $\lambda_v$-calculus and $\pi$-calculus.
Interestingly, in contrast with the parallel encoding of $\lambda$-calculus into $\pi$-calculus,
the parallel encoding of $SF$-calculus into CPC does {\em not} fix a reduction strategy for $SF$-calculus.
This is achieved by exploiting the intensionality of CPC to directly encode the reduction rules for $SF$-calculus into an $SF$-reducing process, or $SF$-machine.
In turn, this process can then operate on translated combinators and so support reduction and rewriting.

The square is completed by showing a homomorphism from $\pi$-calculus into CPC, and by showing that there cannot be any homomorphism (or indeed a more general valid encoding) from CPC into $\pi$-calculus.

\subsubsection*{$SF$-calculus.}
\label{sssec:SF2CPC}

The $SF$-calculus combinators  can be easily encoded into patterns by defining the {\em construction} $\ytrans\cdot$, exploiting reserved names $S$ and $F$, as follows
\begin{eqnarray*}
\ytrans S \define S\qquad
\ytrans F \define F\qquad
\ytrans {M N} \define \ytrans M \bullet\ytrans N\; .
\end{eqnarray*}
Observe that the first two rules map the operators to the same names. The third rule maps application to a compound of the components $\ytrans M$ and $\ytrans N$.

By representing $SF$-calculus combinators in the pattern of a CPC case,
the reduction is driven by cases that recognise a reducible structure and perform the appropriate operations.
The reduction rules can be captured by matching on the structure of the left hand side of the rule and reducing to the structure on the right. So (considering each possible instance for the $F$ reduction rules) they can be encoded by cases as follows
\small
\begin{eqnarray*}
S\bullet \lambda m\bullet \lambda n\bullet \lambda x &\to& m\bullet x\bullet (n\bullet x)\\
F\bullet S\bullet \lambda m\bullet \lambda n &\to& m\\
F\bullet F\bullet \lambda m\bullet \lambda n &\to& m\\
\dots\\
F\bullet (F\bullet\lambda p\bullet \lambda q)\bullet \lambda m\bullet \lambda n &\to&n\bullet (F\bullet p)\bullet q\; .
\end{eqnarray*}
\normalsize
These processes capture the reduction rules, matching the pattern for the left hand side and transforming it to the structure on the right hand side.
Of course these process do not capture the possibility of reduction of a sub-combinator, so further rules are required.
Rather than detail them all, consider the example of a reduction $MNOP \rew MN'OP$ that can be captured by
\begin{eqnarray*}
\lambda m\bullet (\lambda u\bullet \lambda v\bullet \lambda w\bullet \lambda x) \bullet \lambda o\bullet \lambda p
\to u\bullet v\bullet w\bullet x \to \lambda z \to m \bullet z\bullet o\bullet p
\end{eqnarray*}
This process unifies with a combinator $MXOP$ where $X$ is reducible (observable from the structure), here binding the components of $X$ to four names $u$, $v$, $w$ and $x$. These four names are then shared as a pattern, which can then be unified with another process that can perform the reduction. The result will then (eventually) unify with $\lambda z$ and be substituted back into $m\bullet z\bullet o\bullet p$ to complete the reduction.

\begin{figure}[t]
\begin{equation*}
\begin{array}{cl}
& !\lambda c\bullet(S\bullet \lambda m\bullet \lambda n\bullet \lambda x) \to c\bullet(m\bullet x\bullet (n\bullet x))\\
|&!\lambda c\bullet(F\bullet S\bullet \lambda m\bullet \lambda n) \to c\bullet m\\
|&!\lambda c\bullet(F\bullet F\bullet \lambda m\bullet \lambda n) \to c\bullet m\\
|&!\lambda c\bullet(F\bullet (S\bullet \lambda q)\bullet \lambda m\bullet \lambda n) \to c\bullet(n\bullet S\bullet q)\\
|&!\lambda c\bullet(F\bullet (F\bullet \lambda q)\bullet \lambda m\bullet \lambda n) \to c\bullet(n\bullet F\bullet q)\\
|&!\lambda c\bullet(F\bullet (S\bullet\lambda p\bullet \lambda q)\bullet \lambda m\bullet \lambda n)
		\to c\bullet(n\bullet (S\bullet p)\bullet q)\\
|&!\lambda c\bullet(F\bullet (F\bullet\lambda p\bullet \lambda q)\bullet \lambda m\bullet \lambda n)
		\to c\bullet(n\bullet (F\bullet p)\bullet q)\\
|&!\lambda c\bullet(\lambda u\bullet \lambda v\bullet\lambda w\bullet \lambda x\bullet\lambda y)\\
&\quad\to\res d d\bullet (u\bullet v\bullet w\bullet x)\to d\bullet \lambda z
		\to c\bullet(z\bullet y)\\
|&!\lambda c\bullet(\lambda m\bullet \lambda n \bullet \lambda o\bullet (\lambda u\bullet \lambda v\bullet \lambda w\bullet \lambda x))\\
&\quad\to\res d d\bullet(u\bullet v\bullet w\bullet x)\to d\bullet\lambda z
			\to c\bullet (m \bullet n\bullet o\bullet z)\\
|&!\lambda c\bullet(\lambda m\bullet \lambda n \bullet (\lambda u\bullet \lambda v\bullet \lambda w\bullet \lambda x)\bullet \lambda p)\\
&\quad\to\res d d\bullet(u\bullet v\bullet w\bullet x)\to d\bullet\lambda z
			\to c\bullet (m \bullet n\bullet z\bullet p)\\	
|&!\lambda c\bullet(\lambda m\bullet (\lambda u\bullet \lambda v\bullet \lambda w\bullet \lambda x) \bullet \lambda o\bullet \lambda p)\\
&\quad\to\res d d\bullet(u\bullet v\bullet w\bullet x)\to d\bullet\lambda z
			\to c\bullet (m \bullet z\bullet o\bullet p)
\end{array}
\end{equation*}
\caption{The $SF$-reducing process ${\mathcal R}$.}
\label{fig:SF-red-paper}
\end{figure}

To exploit these processes in constructing a parallel encoding requires the addition of a name,
used like a channel, to control application. Thus, prefix each pattern that matches the structure
of an $SF$-combinator with a binding name $\lambda c$ and add this to the result, e.g.
$\lambda c\bullet (F\bullet S\bullet \lambda m\bullet \lambda n) \to c\bullet m$.
Now the processes that handle each possible reduction rule can be placed under a replication
and in parallel composition with each other. This yields the $SF$-reducing process ${\mathcal R}$
as shown in Figure~\ref{fig:SF-red-paper} where the last four replications capture reduction of
sub-combinators.

The translation $\xtrans\cdot _c$ from $SF$-combinators into CPC processes is here parametrised by a name $c$ and combines application with a process ${\sf ap}(c,m,n)$.
This is similar to Milner's encoding from $\lambda_v$-calculus into $\pi$-calculus
and
 allows the parallel encoding to exploit compositional encoding of sub-terms as processes and thus parallel reduction, while preventing confusion of application.

The translation $\xtrans\cdot _c$ of $SF$-combinators into CPC, exploiting the $SF$-reducing process ${\mathcal R}$ and reserved names $S$ and $F$, is defined as follows:
\begin{equation*}
\begin{array}{c}
\xtrans S _c \define c\bullet S \bnf {\mathcal R}\qquad\qquad
\xtrans F _c \define c\bullet F \bnf {\mathcal R}\\
\xtrans {MN} _c \define \res m\res n ({\sf ap}(c,m,n)\bnf \xtrans M _m \bnf \xtrans N _n )\\
{\sf ap}(c,m,n) \define m\bullet \lambda x\to n\bullet \lambda y\to c \bullet (x\bullet y)\bnf  {\mathcal R}\; .
\end{array}
\end{equation*}

The following lemmas are at the core of the operational correspondence and divergence
reflection components of the proof of valid encoding, similar to Milner's Theorem~7.7 \cite{90426}.
Further, it provides a general sense of how to capture the reduction of combinatory logics or
similar rewrite systems.
(Note that the results exploit that ${\mathcal R}\bnf {\mathcal R}\bisim\ {\mathcal R}$ to
remove redundant copies of ${\mathcal R}$ \cite[Theorem~8.7.2]{GivenWilsonPHD}.)

\begin{lemma}
\label{lem:SF2CPC-construct}
Given an $SF$-combinator $M$ the translation $\xtrans M _c$ has a reduction sequence to a process of the form $c\bullet\ytrans M\bnf {\mathcal R}$.
\end{lemma}
\begin{proof}
The proof is by induction on the structure of $M$.
\end{proof}

\begin{lemma}[Theorem~7.1.2 of \cite{GivenWilsonPHD}]
\label{lem:SF2CPC-red}
Given an $SF$-combinator $M$ the translation $\xtrans M _c$ preserves reduction.
\end{lemma}
\begin{proof}
The proof is routine 
by considering each reduction rule and Lemma~\ref{lem:SF2CPC-construct}.
\end{proof}

\begin{lemma}
\label{lem:sf2cpc-red}
The translation $\encode \cdot _c$ preserves and reflects reduction. That is:
\begin{enumerate}
\item If $M\redar N$ then $\encode M _c \Redar\bisim \encode N _c$;
\item if $\encode M _c \redar Q$ then there exists $Q'$ and $N$ such that
      $Q\Redar Q'$ and $Q'\bisim \encode N _c$ and either $M\redar N$ or $M=N$.
\end{enumerate}
\end{lemma}
\begin{proof}
The first part can be proved by exploiting Lemmas~\ref{lem:SF2CPC-construct} and \ref{lem:SF2CPC-red}.
The second is by considering the reduction $\encode M _c\redar Q$ which must arise
from the encoding of an application.
It is then straightforward to show that either:
the reductions $Q\Redar Q'$ correspond only to rebuilding the structure as in
Lemma~\ref{lem:SF2CPC-construct}; or
the reductions correspond to a reduction $M\redar N$ and $Q'\bisim\encode N _c$.
\end{proof}

\begin{theorem}
\label{thm:SF2CPC}
The translation $\xtrans\cdot _c$ is a parallel encoding from $SF$-calculus to CPC.
\end{theorem}
\begin{proof}
Compositionality, parallelisation, and name invariance hold by construction.
Operational correspondence follows from Lemma~\ref{lem:SF2CPC-red}.
Divergence reflection can be proved by observing that the only reductions introduced
in the translation that do not correspond to reductions in the source language are
from translated applications, and these are bounded by the size of the source term.
\end{proof}

The lack of an encoding of CPC (or even $\pi$-calculus) into $SF$-calculus can be
proved in the same manner as Theorem~\ref{thm:no-pi2lambda} for showing no encoding
of $\pi$-calculus into $\lambda$-calculus.

\begin{theorem}
\label{thm:noCPC2SF}
There is no reduction preserving encoding from CPC into $SF$-calculus.
\end{theorem}

It may appear that the factorisation operator $F$ adds some expressiveness that could
be used to capture the parallel-or function $g$.
Perhaps use $F$ to switch on the result of the first function so that (assuming true
is some operator {\tt T} then) $g(x,y)$ is represented by $Fx{\tt T}(K(Ky))$ that reduces to
${\tt T}$ when $x={\tt T}$ and to $K(Ky)MN\Redar y$ when $x=MN$ that somehow is factorable but
not terminating.
However, this kind of attempt is equivalent to exploiting factorisation to detect termination 
and turns out to be paradoxical as demonstrated in the proof of Theorem~5.1 of \cite{JayGW11}.

This completes the arrow down the right side of the computation square. The rest of this section discusses some properties of translations and the diagonal from the top left to the bottom right corner of the square.

\medskip

Observe that the parallel encoding from $SF$-calculus into CPC does not require the choice of a reduction strategy, unlike Milner's encodings from $\lambda$-calculus into $\pi$-calculus.
The structure of patterns and peculiarities of pattern-unification allow the reduction relation to be directly rendered by CPC.
In a sense this is similar to the approach in \cite{givenwilson:hal-00987594} of encoding the $SF$-combinators as the tape
of a Turing Machine, the pattern $\ytrans\cdot$,
and providing another process that reads the tape and performs operations upon it, the $SF$-reducing process
${\mathcal R}$.
This approach can also be adapted in a straightforward manner to support a parallel encoding of
$SK$-calculus into CPC, that like the encoding of $SF$-calculus does not fix a reduction strategy.

\begin{theorem}
\label{thm:SK2CPC}
There is a translation $\xtrans\cdot _c$ that is a parallel encoding from $SK$-calculus into CPC.
\end{theorem}

The translation from $SF$-calculus to CPC presented here is designed to map application to 
parallel composition (with some restriction and process $R$) so as to meet the
compositionality
and parallelisation
criteria for a parallel encoding.
However, the construction $\ytrans\cdot$ can be used to provide a cleaner translation if 
these are not required
(while still supporting the other criteria).
Consider an alternative translation $\xtrans \cdot ^c$ parametrised by a name $c$ as usual and defined by
$\xtrans M ^c \define c\bullet\ytrans M \bnf {\mathcal R}$.

\subsubsection*{$\pi$-calculus.}
\label{subsec:pi}

Across the bottom of the computation square there is a homomorphism from $\pi$-calculus into CPC.
The converse separation result can be proved multiple ways \cite{GivenWilsonGorlaJay10,GivenWilsonPHD,givenwilson:hal-00987578}.

The translation $\encode \cdot$ from $\pi$-calculus into CPC is homomorphic on all process forms except for the input and output which are translated as follows:
\begin{equation*}
\xtrans {\iap a b .P} \define a \bullet \lambda b\bullet \n \to \xtrans P\qquad\qquad
\xtrans {\oap a b .P} \define a\bullet b\bullet \lambda \n\to \xtrans P
\end{equation*}
Here $\n$ is a fresh name (due to the renaming policy 
to avoid all other
names in the translation) that prevents the introduction of new reductions due to CPC's unification.

\begin{lemma}[Corollary~7.2.3 of \cite{GivenWilsonPHD}]
\label{lem:pi2cpc-valid}
The translation $\encode\cdot$ from $\pi$-calculus into CPC is a valid encoding.
\end{lemma}

\begin{theorem}
\label{thm:pi2cpc-hom}
There is a homomorphism (Definition~\ref{def:homo-con}) from $\pi$-calculus into CPC.
\end{theorem}

Thus the translation provided above is a homomorphism from $\pi$-calculus into CPC.
Now consider the converse separation result.

\begin{lemma}[Theorem~7.2.5 of \cite{GivenWilsonPHD}]
\label{lem:nocpc2pi-1}
There is no valid encoding of CPC into $\pi$-calculus.
\end{lemma}
\begin{proof}[Sketch]
Define the {\em self-reducing} CPC process $P=n\to \ok$.
Observe that $P\not\suc$ and $P\bnf P\suc$.
However, for every $\pi$-calculus process $T$ such that $T\bnf T\suc$ it holds that $T\suc$.
This is sufficient to show contradiction of any possible valid encoding.
\end{proof}

\begin{theorem}
\label{thm:noCPC2pi}
There is no homomorphism (Definition~\ref{def:homo-con}) from CPC into $\pi$-calculus.
\end{theorem}

\section{Conclusions and Future Work}
\label{sec:conclusions}

This work illustrates that there are increases in expressive power by shifting along two
dimensions from: extensional to intensional, and sequential to concurrent. This is
seen in
the computation square relating $\lambda_v$-calculus, $SF$-calculus,
$\pi$-calculus, and CPC 
\vspace*{-0.4cm}
\begin{center}
\begin{picture}(250,65)(0,0)
\put(5,43){\mbox{$\lambda_v$-calculus}}
\put(125,43){\mbox{\ \ \ \ \ \ \ \ \ \ \ \ \ $SF$-calculus}}
\put(5,5){\mbox{~$\pi$-calculus}}
\put(123,5){\mbox{concurrent pattern calculus}}
\put(68,45){\vector(1,0){80}}
\put(68,8){\vector(1,0){45}}
\put(34,35){\line(0,-1){5}}
\put(34,26){\vector(0,-1){8}}
\put(183,35){\line(0,-1){5}}
\put(183,26){\vector(0,-1){8}}
\end{picture}
\end{center}
\vspace*{-0.2cm}
where the left side is extensional,
the right side intensional,
the top side sequential,
and the bottom side concurrent.
The horizontal arrows are homomorphisms that map application/parallel composition to itself. 
The vertical arrows are parallel encodings that map application to parallel composition
(with some extra machinery). 
Further, there are no reverse arrows as each arrow signifies an increase in expressive power.

Such a square identifies relations that are more general than simply the choice of calculi here.
The top left corner could be populated by $\lambda_v$-calculus or $\lambda_l$-calculus with minimal changes to the proofs. Alternatively, choosing $\lambda$-calculus or $SK$-calculus may also hold, although a parallel encoding into $\pi$-calculus requires some work.
The top right corner could be populated by any of the structure complete combinatory logics 
\cite{JayGW11,GivenWilsonPHD}. It may also be possible to place a pattern calculus \cite{JK09,pcb}, at the top right.
The bottom left corner is also open to many other calculi: monadic/polyadic synchronous/asynchronous $\pi$-calculus could replace $\pi$-calculus with no significant changes to the results \cite{GivenWilsonPHD,givenwilson:hal-00987578}. 
Similarly there are, and will be, other process calculi that can take the place of CPC at the bottom right.
For Spi calculus \cite{gordon1997ccp} an encoding of $SF$-calculus is delicate due to correctly handling reduction and not introducing infinite reductions or blocking on Spi calculus primitives and reductions. For Psi calculi \cite{BJPV11} the encoding can be achieved very similarly to CPC, although the implicit computation component of Psi calculi could simply allow for $SF$-calculus with the rest being moot.
Although multiple process calculi may populate the bottom right hand corner,
the elegance of CPC's intensionality is illustrated by the construction $\ytrans\cdot$ for combinatory logics and \cite{givenwilson:hal-00987594}.

\withrw{
\vspace*{-0.3cm}
\subsubsection*{Related Work}

The choice of relations here is influenced by existing approaches.
Homomorphisms in the sequential setting are typical \cite{Curry58combinatorylogic,Barendregt85,Felleisen199135}.
Valid encodings are popular \cite{G:IC08,G:DC10,G:CONCUR08,LPSS10,LVF10,gla12}
albeit not the only approach as other ways to relate process calculi are also used
that vary on the choice to map parallel composition to parallel composition (i.e.~homomorphism here)
\cite{journals/corr/cs-PL-9809008,journals/iandc/BusiGZ00,DeNicola:2006:EPK:1148743.1148750,DBLP:journals/corr/abs-1011-6436,gla12}.
Since the choice here is to build on prior results, valid encodings are the obvious
basis 
but no doubt this could be formalised
under different criteria.
Finally, the definition of parallel encodings here is to exploit the existing
encodings in the literature. 
However, other approaches are possible
\cite{90426,parrow.victor:fusion-calculus} and
many more as encoding $\lambda$-calculus into process calculi is common
\cite{96717,milner.parrow.ea:calculus-mobile,1998:MobileAmbients,citeulike:500640}.

The separation results here build upon results already in the literature.
For showing the inability to encoding concurrent languages into sequential, the
work of Abramsky \cite{Abramsky90thelazy} and Plotkin \cite{Plotkin97fullabstraction} can also be considered.
The impossibility of encoding CPC into $\pi$-calculus can be proved by using
matching degree or symmetry \cite[proofs for Theorem~7.2.5]{GivenWilsonPHD}.

}

\vspace*{-0.3cm}
\subsubsection*{Future Work}
Future work may proceed along several directions.
The techniques used to encode $SF$-calculus (here) and Turing Machines \cite{givenwilson:hal-00987594}
into CPC can be generalised for any combinatory logic,
indeed perhaps a general result can be proved for all similar rewrite systems.
Another path of exploration is to consider intensionality in concurrency with full results in a general manner,
this could include formalising the intensionality (or lack of) of Spi calculus, Psi calculi, and other popular process calculi.

\vspace*{-0.3cm}
\bibliographystyle{abbrv}
\bibliography{short}

\end{document}